\documentclass[11pt,letterpaper]{amsart}

\usepackage{amssymb}
\usepackage{enumerate}

\newtheorem{theorem}{Theorem}
\newtheorem{proposition}[theorem]{Proposition}

\newtheorem{lemma}[theorem]{Lemma}

\theoremstyle{remark}

\newcommand{\Z}{\mathbb{Z}}

\newcommand{\R}{\mathbb{R}}

\newcommand{\eps}{\epsilon}

\newcommand{\abs}[1]{\lvert#1\rvert}
\newcommand{\bigabs}[1]{\big\lvert#1\big\rvert}
\newcommand{\biggabs}[1]{\bigg\lvert#1\bigg\rvert}

\newcommand{\Biggabs}[1]{\Bigg\lvert#1\Bigg\rvert}

\newcommand{\sums}[1]{\sum_{\substack{#1}}}

\DeclareMathOperator{\cdf}{CDF}
\DeclareMathOperator{\adf}{ADF}
\DeclareMathOperator{\psc}{PSC}
\DeclareMathOperator{\erfc}{erfc}

\begin{document}

\title[Aperiodic correlation of sequence pairs]{Sequence pairs with asymptotically\\optimal aperiodic correlation}

\author{Christian G\"unther}

\author{Kai-Uwe Schmidt}
\address{Department of Mathematics, Paderborn University, Warburger Str.\ 100, 33098 Paderborn, Germany.}

\email{chriguen@math.upb.de, kus@math.upb.de}

\thanks{The authors are supported by German Research Foundation (DFG)}

\date{22 March 2018 (revised 13 March 2019)}

\begin{abstract}
The Pursley-Sarwate criterion of a pair of finite complex-valued sequences measures the collective smallness of the aperiodic autocorrelations and the aperiodic crosscorrelations of the two sequences. It is known that this quantity is always at least $1$ with equality if and only if the sequence pair is a Golay pair. 
We exhibit pairs of complex-valued sequences whose entries have unit magnitude for which the Pursley-Sarwate criterion tends to $1$ as the sequence length tends to infinity. Our constructions use different carefully chosen Chu sequences.
\end{abstract}

\maketitle

\thispagestyle{empty}


\section{Introduction and Results}

In this paper, a \emph{sequence} of length $n$ is a finite sequence of $n$ complex numbers; such a sequence is called \emph{binary} if each entry is $-1$ or $1$ and it is called \emph{unimodular} if each entry has unit magnitude.
\par
Let $A=(a_0,a_1,\dots,a_{n-1})$ and $B=(b_0,b_1,\dots,b_{n-1})$ be two sequences of length~$n$, where $A$ und $B$ contain at least one nonzero entry. The \emph{aperiodic crosscorrelation} of $A$ and $B$ at shift $u$ is defined to be
\[
C_{A,B}(u)=\sum_{j\in\Z}a_j\overline{b_{j+u}},
\]
with the convention that $a_j=b_j=0$ if $j<0$ or $j\ge n$. The \emph{aperiodic autocorrelation} of $A$ at shift $u$ is then $C_{A,A}(u)$. 
\par
There is sustained interest in sequences with small correlation (see~\cite{Sch2016} for a survey of recent developments), mainly because small correlation helps to separate a useful signal from noise or unwanted signals. In particular, small crosscorrelation is usually required to ensure that sequences can be distinguished well from each other and small autocorrelation is usually required to keep the transmitter and the receiver synchronised.
\par
The collective smallness of the aperiodic crosscorrelations of the sequences $A$ and $B$ is measured by the \emph{crosscorrelation demerit factor} of $A$ and~$B$, which is defined to be
\[
\cdf(A,B)=\frac{\sum_{u\in\Z}\abs{C_{A,B}(u)}^2}{C_{A,A}(0)\cdot C_{B,B}(0)}.
\]
Accordingly, the collective smallness of the aperiodic autocorrelations of $A$ is measured by the \emph{autocorrelation demerit factor} of $A$, which is defined to be
\[
\adf(A)=\frac{\sum_{u\in\Z\setminus\{0\}}\abs{C_{A,A}(u)}^2}{C_{A,A}(0)^2}.
\]
Notice that $C_{A,A}(0)=C_{B,B}(0)=n$ if $A$ and $B$ are unimodular. The reciprocals of $\cdf(A,B)$ and $\cdf(A)$ are known as the \emph{crosscorrelation merit factor} of $A$ and $B$ and the \emph{autocorrelation merit factor} of $A$, respectively.
\par
A fundamental relationship between the three quantities $\adf(A)$, $\adf(B)$, and $\cdf(A,B)$ is given by
\begin{equation}
\label{eqn:bound_PSC}
-\sqrt{\adf(A)\adf(B)}\le\cdf(A,B)-1\le \sqrt{\adf(A)\adf(B)},
\end{equation}
as proved by Pursley and Sarwate~\cite{PurSar1976} for binary sequences and generalised by Katz and Moore~\cite{KatMoo2017} for general sequences. Following Boothby and Katz~\cite{BooKat2017}, we define the \emph{Pursley-Sarwate criterion} of~$A$ and~$B$ to be
\[
\psc(A,B)=\sqrt{\adf(A)\adf(B)}+\cdf(A,B).
\]
From~\eqref{eqn:bound_PSC} we obtain $\psc(A,B)\ge 1$. Hence in order to design pairs of sequences $(A,B)$ with simultaneously small aperiodic autocorrelations and crosscorrelations, we would like to have $\psc(A,B)$ close to $1$.
\par
Katz~\cite{Kat2016} and Boothby and Katz~\cite{BooKat2017} studied the Pursley-Sarwate criterion of sequence pairs derived from m-sequences and Legendre sequences and generalisations of them. This gives pairs of unimodular and binary sequences whose Pursley-Sarwate criterion is close to $1$, but strictly bounded away from~$1$, as the sequence length tends to infinity. 
\par
In fact, pairs of unimodular sequences $(A,B)$ with $\psc(A,B)=1$ were recently classified by Katz and Moore~\cite{KatMoo2017} to be exactly the \emph{Golay pairs}, namely pairs of unimodular sequences $(A,B)$ satisfying
\[
C_{A,A}(u)+C_{B,B}(u)=0\quad\text{for all $u\ne 0$}.
\]
Golay pairs of unimodular sequences are known to exist for infinitely many, though not for all, lengths (see~\cite{Fie2013} for the existence for small lengths). The classification in~\cite{KatMoo2017} does however not say anything about the individual quantities $\adf(A)$, $\adf(B)$, and $\cdf(A,B)$ for a Golay pair $(A,B)$. These values are known for the \emph{Rudin-Shapiro pairs}, which are Golay pairs $(A_m,B_m)$ of binary sequences of length $2^m$ and satisfy~\cite{Lit1968}
\[
\adf(A_m)=\adf(B_m)=\tfrac{1}{3}(1-(-1/2)^m).
\]
Since $\psc(A_m,B_m)=1$, we have
\[
\cdf(A_m,B_m)=\tfrac{1}{3}(2+(-1/2)^m).
\]
Therefore, as $m\to\infty$, we have $\adf(A_m)\to 1/3$ and $\adf(B_m)\to1/3$ and $\cdf(A_m,B_m)\to2/3$.
\par
In this paper we exhibit unimodular sequences whose Pursley-Sarwate criterion is asymptotically $1$ and for which (unlike for general Golay pairs) we can control the autocorrelation and crosscorrelation demerit factors. Our results involve \emph{Chu} sequences~\cite{Chu1972}, which are unimodular sequences of length~$n$ of the form
\[
Z_n^{(a)}=(z_0,z_1,\dots,z_{n-1}),\quad z_j=e^{\pi iaj^2/n},
\]
where~$a$ is an integer (Chu~\cite{Chu1972} used a slightly different definition when $n$ is odd).
\par
Several authors~\cite{New1965},~\cite{StaBoc2000},~\cite{Mer2013} have shown independently that $1/\adf(Z_n^{(1)})$ grows like~$\sqrt{n}$ and the exact constant has been determined by the second author~\cite{Sch2013} by showing that
\begin{equation}
\lim_{n\to\infty}\sqrt{n}\,\adf(Z_n^{(1)})=\frac{2}{\pi},   \label{eqn:adf_chu}
\end{equation}
which confirms previous numerical evidence obtained by Littlewood~\cite{Lit1966} and B\"omer and Antweiler~\cite{AntBom1990}. In fact, since $\adf(Z_n^{(1)})$ tends to zero, this immediately implies that
\[
\lim_{n\to\infty} \psc(Z_n^{(1)},Z_n^{(1)})=1.
\]
However, the pair $(Z_n^{(1)},Z_n^{(1)})$ would be a bad choice when good crosscorrelation is required since the crosscorrelation at the zero shift equals the sequence length $n$. This problem is avoided by taking the pair $(Z_n^{(1)},Z_n^{(-1)})$. Indeed it can be shown using Lemma~\ref{lem:bound_on_generalised_gauss_sum} that in this case the crosscorrelation at the zero shift is of the order $\sqrt{n}$. Our first result is the following theorem.
\begin{theorem}
\label{thm:main2}
For each $n$, let $X_n=Z_n^{(1)}$ and $Y_n=Z_n^{(-1)}$ be Chu sequences of length $n$. Then, as $n\to\infty$,
\begin{enumerate}[(i)]
\item $\adf(X_n)\to 0$, $\adf(Y_n)\to 0$,
\item $\cdf(X_n,Y_n)\to 1$,
\item $\psc(X_n,Y_n)\to 1$.
\end{enumerate}
\end{theorem}
\par
In our second result we construct a pair of unimodular sequences from two Chu sequences of even length such that the Pursley-Sarwate criterion is asymptotically $1$ and the autocorrelation and crosscorrelation demerit factors are asymptotically balanced, which means that they all tend to the same constant~$1/2$.
\begin{theorem}
\label{thm:main}
For each $n$, let $X_n=Z_{2n}^{(n+1)}$ and $Y_n=Z_{2n}^{(n-1)}$ be Chu sequences of length $2n$. Then, as $n\to\infty$,
\begin{enumerate}[(i)]
\item $\adf(X_n)\to 1/2$, $\adf(Y_n)\to 1/2$,
\item $\cdf(X_n,Y_n)\to 1/2$,
\item $\psc(X_n,Y_n)\to 1$.
\end{enumerate}
\end{theorem}
\par
We shall prove the first parts of Theorems~\ref{thm:main2} and~\ref{thm:main} in Section~\ref{sec:acf} and the second parts in Section~\ref{sec:ccf}. Of course the third parts are trivial consequences of the first two parts.
\par
We close this section with some clarifying comments on the relationship of~\cite{Kat2016} with earlier work, whose detailed discussion in this paper was requested by one of the referees. The initial version of this paper contained the side remark ``Unlike claimed in~\cite{Kat2016}, the crosscorrelation demerit factors of pairs obtained from m-sequences and Legendre sequences have been studied before, namely in~\cite[Theorem~3]{SchJedPar2009} and~\cite[Theorem 9]{JedSch2010}.'' As indicated by the referee, this statement is incorrect; as explained in~\cite[Section~II.C]{Kat2016}, Sarwate~\cite{Sar1984} studied the crosscorrelation demerit factors of pairs of m-sequences. Our remark was an attempt to comment on the following statement of~\cite[p.~5238]{Kat2016}: ``\dots relatively little is known about the more difficult problem of characterizing the crosscorrelation of important sequence families''. Indeed special cases of~\cite[Theorems 3 and 4]{Kat2016}, which are the main results of~\cite{Kat2016}, were previously obtained in~\cite[Theorem~9]{JedSch2010} (for m-sequences) and~\cite[Theorem~3]{SchJedPar2009} (for Legendre sequences) and these earlier results were not mentioned in~\cite{Kat2016}. To the best of our knowledge,~\cite[Theorem~3]{SchJedPar2009} contains the first results on the crosscorrelation demerit factors of distinct sequences obtained from multiplicative characters of finite fields. We note however that in the results of~\cite[Theorem~9]{JedSch2010} and~\cite[Theorem~3]{SchJedPar2009} (unlike in the more general results of~\cite[Theorems 3 and 4]{Kat2016}), one sequence in the pair is a cyclic shift of the other sequence in the pair, which means that one value of the crosscorrelation is at least as large as half the sequence length. This makes these pairs bad choices in practice.


\section{Autocorrelation of Chu sequences}
\label{sec:acf}

In this section we prove our results on the autocorrelations of the Chu sequences in question. We begin with a lemma, which gives an expression for the autocorrelation demerit factor of Chu sequences.
\begin{lemma}
\label{lem:mf1}
Let $a$ and $n\ge 1$ be integers and write $d=\gcd(a,n)$. Then 
\[
\adf(Z_n^{(a)})=\frac{4d}{n^2}\sum_{1\leq u\le n/(2d)}\bigg(\frac{\sin(\pi a u^2/n)}{\sin(\pi au/n)}\bigg)^2+\frac{(d-1)(2d-1)}{3d}-\frac{2d\eps_{n/d}}{n^2},
\]
where $\eps_{n/d}=1$ if ${n/d}\equiv 2\pmod 4$ and $\eps_{n/d}=0$ otherwise.
\end{lemma}
\begin{proof}
Write $Z=Z_n^{(a)}$, $a=db$, and $n=dm$. Straightforward manipulations give
\[
\abs{C_{Z,Z}(u)}=\abs{C_{Z,Z}(-u)}=\Biggabs{\sum_{j=0}^{n-u-1}e^{2\pi i b u j/m}}
\]
for $0\le u<n$. Hence, if $u$ is not a multiple of $m$, then
\[
\abs{C_{Z,Z}(u)}=\biggabs{\frac{\sin(\pi b u^2/m)}{\sin(\pi b u/m)}}.
\]
Therefore
\begin{align*}
\adf(Z_n^{(a)})&=\frac{2}{n^2}\sum_{u=1}^{n-1}\abs{C_{Z,Z}(u)}^2\\
&=\frac{2d}{n^2}\sum_{u=1}^{m-1}\abs{C_{Z,Z}(u)}^2+\frac{2}{n^2}\sum_{k=1}^{d-1}\abs{C_{Z,Z}(km)}^2\\
&=\frac{4d}{n^2}\sum_{1\le u\le m/2}\abs{C_{Z,Z}(u)}^2+\frac{2}{n^2}\sum_{k=1}^{d-1}\abs{C_{Z,Z}(km)}^2-\frac{2d\eps_m}{n^2}
\end{align*}
since $\abs{C_{Z,Z}(u)}=\abs{C_{Z,Z}(m-u)}$ for $1\le u<m$ and $\abs{C_{Z,Z}(m/2)}=\eps_m$ for even $m$. Finally, note that
\begin{align*}
\sum_{k=1}^{d-1}\abs{C_{Z,Z}(km)}^2&=\sum_{k=1}^{d-1}(n-km)^2\\
&=m^2\sum_{k=1}^{d-1}k^2\\
&=\frac{n^2}{6d}(d-1)(2d-1),
\end{align*}
which completes the proof.
\end{proof}
\par
The second author obtained~\eqref{eqn:adf_chu} from Lemma~\ref{lem:mf1} with $a=1$ and the first identity in the following lemma.
\begin{lemma}[\cite{Sch2013}]
\label{lem:mf_chu_schmidt}
We have
\begin{align*}
\lim_{n\to\infty}\frac{1}{n^{3/2}}\sum_{1\leq u \leq n/2}\bigg(\frac{\sin(\pi u^2/n)}{\sin(\pi u/n)}\bigg)^2&=\frac{1}{2\pi},\\
\lim_{n\to\infty}\frac{1}{n^{3/2}}\sum_{1\leq u \leq n/2}\bigg(\frac{\sin(\pi u^2/n)}{\pi u/n}\bigg)^2&=\frac{1}{2\pi}.
\end{align*}
\end{lemma}
\par
From Lemmas~\ref{lem:mf1} and~\ref{lem:mf_chu_schmidt} we also obtain 
\[
\lim_{n\to\infty} n^{1/2}\adf(Z_n^{(-1)})=\frac{2}{\pi},
\]
which implies Theorem~\ref{thm:main2}~(i). We now use Lemmas~\ref{lem:mf1} and~\ref{lem:mf_chu_schmidt} to prove Theorem~\ref{thm:main}~(i).
\begin{proof}[Proof of Theorem~\ref{thm:main}~(i)]
We distinguish the cases that $n$ runs through the set of even and odd positive integers. First we show that
\begin{equation}
\lim_{m\to\infty}\adf(X_{2m})=\lim_{m\to\infty}\adf(Y_{2m})=1/2.   \label{eqn:mf_2m}
\end{equation}
It is readily verified that the aperiodic autocorrelations of $X_{2m}$ and $Y_{2m}$ have equal magnitudes at all shifts, so that it is sufficient to establish that $\adf(X_{2m})\to 2$. Noting that $2m+1$ is coprime to $4m$ and using trigonometric addition formulas, Lemma~\ref{lem:mf1} with $n=4m$ and $a=2m+1$ shows that $4m^2\adf(X_{2m})$ equals
\[
\sums{u=1\\\text{$u$ even}}^{2m}\bigg(\frac{\sin(\pi u^2/(4m))}{\sin(\pi u/(4m))}\bigg)^2+\sums{u=1\\\text{$u$ odd}}^{2m}\bigg(\frac{\cos(\pi u^2/(4m))}{\cos(\pi u/(4m))}\bigg)^2.
\]
By Lemma~\ref{lem:mf_chu_schmidt}, the first sum is $O(m^{3/2})$, so that it is sufficient to show that
\begin{equation}
\label{eqn:proof_of_thm_main_cos}
\lim_{m\to\infty}\frac{1}{m^2}\sum_{v=1}^m\bigg(\frac{\cos((\pi (2v-1)^2/(4m))}{\cos(\pi (2v-1)/(4m))}\bigg)^2=2.
\end{equation}
Let $x$ be a real number satisfying $0<x< \pi/2$. From the Taylor series of $\cos x$ at $\pi/2$ we find that
\[
-(x-\pi/2)+\tfrac{1}{6}(x-\pi/2)^3<\cos x<-(x-\pi/2),
\]
from which it follows that
\[
0<\frac{1}{(\cos x)^2}-\frac{1}{(x-\pi/2)^2}<1.
\]
Therefore,
\[
\Biggabs{\sum_{v=1}^m\bigg(\frac{\cos(\pi (2v-1)^2/(4m))}{\cos(\pi (2v-1)/(4m))}\bigg)^2-\sum_{v=1}^m\bigg(\frac{\cos(\pi (2v-1)^2/(4m))}{\pi (2v-1)/(4m)-\pi/2}\bigg)^2\,}<m.
\]
Put $v=m+1-w$ to obtain
\[
\sum_{v=1}^m\bigg(\frac{\cos(\pi (2v-1)^2/(4m))}{\pi (2v-1)/(4m)-\pi/2}\bigg)^2=\sum_{w=1}^m\bigg(\frac{\cos(\pi (2w-1)^2/(4m))}{\pi(2w-1)/(4m)}\bigg)^2.
\]
Apply $(\cos x)^2=1-(\sin x)^2$ to the right hand side and use Lemma~\ref{lem:mf_chu_schmidt} to obtain
\begin{align*}
\lim_{m\to\infty}\frac{1}{m^2}\sum_{w=1}^m\bigg(\frac{\cos(\pi (2w-1)^2/(4m))}{\pi(2w-1)/(4m)}\Bigg)^2&=\frac{16}{\pi^2}\sum_{w=1}^{\infty}\frac{1}{(2w-1)^2}\\
&=\frac{16}{\pi^2}\Bigg(\sum_{w=1}^\infty\frac{1}{w^2}-\sum_{w=1}^\infty\frac{1}{(2w)^2}\Bigg)\\
&=\frac{12}{\pi^2}\sum_{w=1}^\infty\frac{1}{w^2}=2,
\end{align*}
using Euler's evaluation $\sum_{w=1}^\infty1/w^2=\pi^2/6$. This gives~\eqref{eqn:proof_of_thm_main_cos}, and so completes the proof of~\eqref{eqn:mf_2m}.
\par
Next we prove that
\begin{equation}
\lim_{m\to\infty}\adf(X_{2m+1})=\lim_{m\to\infty}\adf(Y_{2m+1})=1/2.   \label{eqn:mf_2m1}
\end{equation}
Note that $\gcd(2m+2,4m+2)=2$, so that Lemma~\ref{lem:mf1} with $n=4m+2$ and $a=2m+2$ gives
\[
\adf(X_{2m+1})=\frac{8}{(4m+2)^2}\sum_{u=1}^m\Bigg(\frac{\sin(\pi (m+1) u^2/(2m+1))}{\sin(\pi (m+1) u/(2m+1))}\bigg)^2+\frac{1}{2}.
\]
On the other hand, $\gcd(2m,4m+2)=2$, so that Lemma \ref{lem:mf1} with $n=4m+2$ and $a=2m$ gives
\[
\adf(Y_{2m+1})=\frac{8}{(4m+2)^2}\sum_{u=1}^m\bigg(\frac{\sin(\pi m u^2/(2m+1))}{\sin(\pi m u/(2m+1))}\bigg)^2+\frac{1}{2}.
\]
Comparing the two preceding equations and using
\[
\biggabs{\sin\bigg(\frac{\pi m k}{2m+1}\bigg)}=\biggabs{\sin\bigg(\frac{\pi m k}{2m+1}-\pi k\bigg)}=\biggabs{\sin\bigg(\frac{\pi(m+1)k}{2m+1}\bigg)},
\]
we conclude that $\adf(X_{2m+1})=\adf(Y_{2m+1})$. To complete the proof, we show that 
\begin{equation}
\label{eqn:proof_of_thm_main_n_odd}
\lim_{m\to\infty}\frac{1}{m^2}\sum_{u=1}^m\bigg(\frac{\sin(\pi m u^2/(2m+1))}{\sin(\pi m u/(2m+1))}\bigg)^2=0.
\end{equation}
For even $k$, we have
\[
\biggabs{\sin\bigg(\frac{\pi m k}{2m+1}\bigg)}=\biggabs{\sin\bigg(\frac{\pi m k}{2m+1}-\frac{\pi}{2}k\bigg)}=\biggabs{\sin\bigg(\frac{\pi k}{4m+2}\bigg)},
\]
so that 
\[
\sums{u=1\\\text{$u$ even}}^m\bigg(\frac{\sin(\pi m u^2/(2m+1))}{\sin(\pi m u/(2m+1))}\bigg)^2=O(m^{3/2})
\]
by Lemma \ref{lem:mf_chu_schmidt}. For odd $k$, we have
\[
\biggabs{\sin\bigg(\frac{\pi m k}{2m+1}\bigg)}=\biggabs{\cos\bigg(\frac{\pi m k}{2m+1}-\frac{\pi}{2}k\bigg)}=\biggabs{\cos\bigg(\frac{\pi k}{4m+2}\bigg)},
\]
so that
\[
\sums{u=1\\ \text{$u$ odd}}^m\bigg(\frac{\sin(\pi m u^2/(2m+1))}{\sin(\pi m u/(2m+1))}\bigg)^2=\sums{u=1\\ \text{$u$ odd}}^m\bigg(\frac{\cos(\pi u^2/(4m+2))}{\cos(\pi u/(4m+2))}\bigg)^2.
\]
The summands on the right hand side are at most $2$, so that the entire sum is at most $m+1$. This proves~\eqref{eqn:proof_of_thm_main_n_odd}, and so completes the proof of~\eqref{eqn:mf_2m1}.
\end{proof}


\section{Crosscorrelation of Chu sequences}
\label{sec:ccf}

In this section we prove our results on the crosscorrelations of the Chu sequences in question. A straightforward computation shows that the aperiodic crosscorrelations between two Chu sequences of equal length are equal in modulus to \emph{generalised Gauss sums}, which are for real $x$ and $\theta$ and integral $N$ defined to be
\[
S_N(x,\theta)=\sum_{j=1}^Ne^{\pi ixj^2+2\pi i\theta j}.
\]
An asymptotic expansion of these sums was obtained by Paris~\cite{Par2014} using an asymptotic expansion of the error function. We deduce an estimate for generalised Gauss sums from this expansion. To state the result, define for real $x\ne 0$ and $\theta$,
\[
E(x,\theta)=e^{-\pi i \theta^2/x}\erfc\big(e^{\pi i/4}\theta\sqrt{\pi/x}\big),
\]
where, for complex $z$,
\[
\erfc(z)=1-\frac{2}{\sqrt{\pi}}\int_0^ze^{-t^2}\,dt
\]
is the complementary error function and the integral is over any path from~$0$ to~$z$.
\begin{proposition}[{\cite[Theorem~1]{Par2014}}]
\label{prop:paris}
Let $N$ be a positive integer, let $x\in (0,1)$, and let $\theta\in (-1/2,1/2]$. Write $Nx+\theta=M+\epsilon$, where $M$ is integral and $\epsilon\in (-1/2,1/2]$. Then
\begin{align*}
S_N(x,\theta)&=\frac{e^{-\pi i\theta^2/x+\pi i/4}}{\sqrt{x}}S_M(-1/x,\theta/x)+\frac{\mu-1}{2}\\
&~+\frac{e^{\pi i/4}}{2\sqrt{x}}\big(E(x,\theta)-\mu E(x,\epsilon)\big)+\frac{i}{2}\big(g(\theta)-\mu g(\epsilon)\big)+R,
\end{align*}
where $\abs{R}<x$ and $\mu=e^{\pi ixN^2+2\pi i\theta N}$ and $g:[-1/2,1/2]\to\R$ is given by
\[
g(t)=\begin{cases}
0 & \text{for $t=0$}\\
\cot(\pi t)-(\pi t)^ {-1} & \text{otherwise}.
\end{cases}
\]
\end{proposition}
\par
Fiedler, Jurkat, and K\"orner~\cite{FieJurKoe1977} obtained the following estimate from a slightly weaker version of Proposition~\ref{prop:paris}.
\begin{lemma}[{\cite[Lemma~4]{FieJurKoe1977}}]
\label{lem:bound_on_generalised_gauss_sum}
Let $N$, $k$, and $m$ be integers such that $\gcd(k,m)$ and $N/m$ are bounded by absolute constants and let $\theta$ be real. Then
\[
\abs{S_N(k/m,\theta)}=O(\sqrt{m}),
\] 
where the implicit constant is absolute.
\end{lemma}
\par
From Proposition~\ref{prop:paris} we can also deduce the following result, which will be the key ingredient to our proofs of Theorem~\ref{thm:main2} and~\ref{thm:main}.
\begin{lemma}
\label{lem:asymptotic_generalised_gauss_sum} 
Let $m$ be a positive integer and let $u$ be an integer such that either $u$ or $m-u$ is in the set
\begin{equation}
\{w\in\Z:m^{2/3}\le w\leq m/2-m^{2/3}\}.   \label{eqn:set_u}
\end{equation}
Then 
\begin{equation}
\bigabs{S_{m-u}(2/m,u/m)}=\sqrt{\frac{m}{2}}+O(m^{1/3}),   \label{eqn:estimate_Gauss}
\end{equation}
where the implicit constant is absolute.
\end{lemma}
\begin{proof}
Throughout the proof, all implicit constants are absolute. We first prove the desired bound when $u$ is in the set~\eqref{eqn:set_u}. This is an application of Proposition~\ref{prop:paris} with $N=m-u$, $x=2/m$, $\theta=u/m$, so that $M=2$ and $\epsilon=-u/m$ and $\mu=1$. We have
\begin{align*}
\cot(\pi\theta)-\cot(\pi\epsilon)-\frac{1}{\pi\theta}+\frac{1}{\pi\epsilon}&=2\cot(\pi u/m)-\frac{2m}{\pi u}\\
&=O(m^{1/3}),
\end{align*}
using that $\cot(\pi t)$ is nonnegative on $(0,1/2]$ and, for $m\ge 8$,
\begin{align*}
\cot(\pi u/m)&\le \cot\big(\pi m^{-1/3}\big)\\
&\le\frac{1}{\sin\big(\pi m^{-1/3}\big)}\\
&\le\frac{1}{2}m^{1/3}.
\end{align*}
Now use the identity
\[
\erfc(-z)=2-\erfc(z), 
\]
to obtain
\[
E(x,\theta)-E(x,-\theta)=2E(x,\theta)-2e^{-\pi i\theta^2/x}.
\]
Since
\[
S_2(-1/x,\theta/x)=(-i)^m(-1)^u+1,
\]
we then find from Proposition~\ref{prop:paris} that
\[
\abs{S_N(x,\theta)}=\sqrt{\frac{m}{2}}\;\bigabs{(-i)^m(-1)^u+E(x,\theta)}+O(m^{1/3}).
\]
From the asymptotic expansion~\cite[(1.3) and (1.4)]{Par2014} of $E(x,\theta)$ we find that
\[
\biggabs{E(x,\theta)-\frac{e^{-2\pi i\theta^2/x-\pi i/4}\sqrt{x}}{\pi\theta}}\le\frac{x^ {3/2}}{2\pi^2\theta^3},
\]
so that $\abs{E(x,\theta)}$ is $O(m^{-1/6})$ and~\eqref{eqn:estimate_Gauss} follows, as required.
\par
Now assume that $m-u$ is in the set~\eqref{eqn:set_u}. Put $v=m-u$ and apply Proposition~\ref{prop:paris} with $N=v$, $x=2/m$, and $\theta=-v/m$, so that $M=0$ and $\epsilon=v/m$ and $\mu=1$. Proceeding similarly as in the first case, we obtain
\[
\abs{S_N(x,\theta)}=\sqrt{\frac{m}{2}}\;\bigabs{e^{-\pi i\theta^2/x}-E(x,\theta)}+O(m^{1/3}),
\]
which again implies~\eqref{eqn:estimate_Gauss}.
\end{proof}
\par
We conclude by proving the second parts of Theorems~\ref{thm:main2} and~\ref{thm:main}.
\begin{proof}[Proof of Theorem~\ref{thm:main2}~(ii)]
For each $u\in\{0,1,\dots,n-1\}$, we have
\[
\abs{C_{X_n,Y_n}(u)}=\abs{C_{X_n,Y_n}(-u)}=\Biggabs{\sum_{j=0}^{n-u-1}e^{2\pi i(j^2+ju)/n}}.
\]
Extend the summation range to $n-u$ and compensate the extra term (which equals $1$) by taking out the term corresponding to $j=0$. The sum then equals $S_{n-u}(2/n,u/n)$ and therefore
\[
n^2\cdf(X_n,Y_n)=\abs{S_n(2/n,0)}^2+2\sum_{u=1}^{n-1}\abs{S_{n-u}(2/n,u/n)}^2.
\]
Now use Lemma~\ref{lem:asymptotic_generalised_gauss_sum} to find that $\abs{S_{n-u}(2/n,u/n)}^2=n/2+O(n^{5/6})$ when $u$ or $n-u$ is in the set
\[
\{w\in\Z:n^{2/3}\le w\le n/2-n^{2/3}\}
\]
and use Lemma~\ref{lem:bound_on_generalised_gauss_sum} to bound the remaining $O(n^{2/3})$ values of $\abs{S_{n-u}(2/n,u/n)}^2$ by $O(n)$. This gives
\[
\cdf(X_n,Y_n)=1+O(n^{-1/6}),
\]
as required.
\end{proof}
\par
\begin{proof}[Proof of Theorem~\ref{thm:main}~(ii)]
For each $u\in\{0,1,\dots,2n-1\}$, we have
\[
\abs{C_{X_n,Y_n}(u)}=\abs{C_{X_n,Y_n}(-u)}=\Biggabs{\sum_{j=0}^{2n-u-1}(-1)^{ju}e^{\pi i(j^2+ju)/n}}.
\]
If~$u$ is odd, then the summands corresponding to $j$ and $2n-u-j$ add to zero, leaving the summand corresponding to $j=0$. Therefore $\abs{C_{X_n,Y_n}(u)}=1$ for each odd $u$ satisfying $\abs{u}\le 2n-1$. Hence
\[
(2n)^2\cdf(X_n,Y_n)=\abs{S_{2n}(1/n,0)}^2+2\sum_{v=1}^{n-1}\abs{S_{2n-2v}(1/n,v/n)}^2+O(n).
\]
Now use Lemma~\ref{lem:asymptotic_generalised_gauss_sum} to find that $\abs{S_{2n-2v}(1/n,v/n)}^2=n+O(n^{5/6})$ when $2v$ or $2n-2v$ is in the set
\[
\{w\in\Z:(2n)^{2/3}\le w\le n-(2n)^{2/3}\}
\]
and use Lemma~\ref{lem:bound_on_generalised_gauss_sum} to bound the remaining $O(n^{2/3})$ values of $\abs{S_{2n-2v}(1/n,v/n)}^2$ by $O(n)$. This gives
\[
\cdf(X_n,Y_n)=\frac{1}{2}+O(n^{-1/6}),
\]
as required.
\end{proof}


\section{Concluding remarks}

We exhibited two families of pairs of unimodular sequences whose Pursley-Sarwate criterion tends to $1$ as the sequence length tends to infinity and determined the asymptotic autocorrelation and crosscorrelation demerit factors of the sequences in these pairs. In particular, the triple
\begin{equation}
\lim_{n\to\infty}\big(\adf(A_n),\adf(B_n),\cdf(A_n,B_n)\big)   \label{eqn:triple}
\end{equation}
equals $(0,0,1)$ for the pair of Chu sequences in Theorem~\ref{thm:main2} and equals $(\tfrac12,\tfrac12,\tfrac12)$ for the pair of Chu sequences in Theorem~\ref{thm:main}. The only other family of pairs of unimodular sequences with asymptotic Pursley-Sarwate criterion equal to $1$ and  for which~\eqref{eqn:triple} is known is the family of Shapiro sequence pairs, for which the triple~\eqref{eqn:triple} equals $(\tfrac13,\tfrac13,\tfrac23)$. We ask: what are the possible triples~\eqref{eqn:triple} for families of pairs of unimodular sequences $(A_n,B_n)$ of increasing length whose Pursley-Sarwate criterion is asymptotically $1$?



\providecommand{\bysame}{\leavevmode\hbox to3em{\hrulefill}\thinspace}
\providecommand{\MR}{\relax\ifhmode\unskip\space\fi MR }
\providecommand{\MRhref}[2]{%
  \href{http://www.ams.org/mathscinet-getitem?mr=#1}{#2}
}
\providecommand{\href}[2]{#2}

\end{document}